\documentclass[onefignum,onetabnum]{siamart190516}
\usepackage{graphicx}
\pdfoutput=1

\newcommand{\incfig}{\centering\includegraphics}
\setkeys{Gin}{width=0.7\linewidth,keepaspectratio}

\newcommand{\eqr}[1]{Eq.\thinspace(#1)}
\newcommand{\pfrac}[2]{\frac{\partial #1}{\partial #2}}

\newcommand{\pfraca}[1]{\frac{\partial}{\partial #1}}

\newcommand{\mvec}[1]{\mathbf{#1}}
\newcommand{\gvec}[1]{\boldsymbol{#1}}

\newcommand{\gcs}{\nabla_{\mvec{x}}}
\newcommand{\gvs}{\nabla_{\mvec{v}}}
\newcommand{\gke}{\texttt{Gkeyll}}

\newsiamremark{remark}{Remark}
\newsiamremark{hypothesis}{Hypothesis}
\crefname{hypothesis}{Hypothesis}{Hypotheses}
\newsiamthm{claim}{Claim}

\begin{document}

\title{Discontinuous Galerkin schemes for a class of Hamiltonian
  evolution equations with applications to plasma fluid and kinetic
  problems}%

\author{A. Hakim\thanks{Princeton Plasma Physics Laboratory, Princeton, NJ 08543-0451, USA} \and G.~W. Hammett\footnotemark[1] \and
  E.~L.~Shi\thanks{Lawrence Livermore National Laboratory, Livermore,
    CA 94550, USA} \and N.~R.~Mandell\thanks{Department of Astrophysical Sciences, Princeton University,
    Princeton, NJ 08544, USA}}%

\maketitle

\begin{center}
{\footnotesize \today} 
\end{center}

\begin{abstract}
  In this paper we present energy-conserving, mixed discontinuous
  Galerkin (DG) and continuous Galerkin (CG) schemes for the solution
  of a broad class of physical systems described by Hamiltonian
  evolution equations. These systems often arise in fluid mechanics
  (incompressible Euler equations) and plasma physics (Vlasov--Poisson
  equations and gyrokinetic equations), for example. The dynamics is
  described by a distribution function that evolves given a
  Hamiltonian and a corresponding Poisson bracket operator, with the
  Hamiltonian itself computed from field equations. Hamiltonian
  systems have several conserved quantities, including the quadratic
  invariants of total energy and the $L_2$ norm of the distribution
  function. For accurate simulations one must ensure that these
  quadratic invariants are conserved by the discrete scheme. We show
  that using a discontinuous Galerkin scheme to evolve the
  distribution function and ensuring that the Hamiltonian lies in its
  continuous subspace leads to an energy-conserving scheme in the continuous-time limit. Further,
  the $L_2$ norm is conserved if central fluxes are used to update the
  distribution function, but decays monotonically when using upwind
  fluxes. The conservation of density and $L_2$ norm is then used to
  show that the entropy is a non-decreasing function of time. The
  proofs shown here apply to any Hamiltonian system, including ones in
  which the Poisson bracket operator is non-canonical (for example,
  the gyrokinetic equations). We demonstrate the ability of the scheme
  to solve the Vlasov--Poisson and incompressible Euler equations in 2D
  and provide references where we have applied these schemes to solve
  the much more complex 5D electrostatic and electromagnetic
  gyrokinetic equations.
\end{abstract}


\section{Introduction}
\label{sec:intro}

In this paper we present energy-conserving, mixed discontinuous
Galerkin (DG)/continuous Galerkin (CG) schemes for the solution of a
broad class of problems in fluid mechanics and plasma physics,
described by a Hamiltonian evolution equation
\begin{align}
  \pfrac{f}{t} + \{f,H\} = 0 \label{eq:fevolve}.
\end{align}
Here $f(t,\mvec{z})$ is a distribution function, $H(\mvec{z})$ is the
Hamiltonian and $\{g,h\}$ is the Poisson bracket operator. 
The
coordinates $\mvec{z}=(z^1,\ldots,z^N)$ label the $N$-dimensional
phase space in which the distribution function evolves. The
Hamiltonian itself is determined from the solution of a field
equation, usually an elliptic (\emph{e.g.} Poisson) or hyperbolic (\emph{e.g.} Maxwell)
partial differential equation in configuration space.  For an overview
of Hamiltonian dynamics, see, for example, the textbook of Sudarshan
and Mukunda \cite{sudarshan74} or section II of Cary and
Brizard \cite{Cary:2009dl}.

Defining the phase-space velocity vector $\gvec{\alpha} =
(\dot{z}^1,\ldots,\dot{z}^N)$, where the characteristic speeds are
determined from $\dot{z}^i = \{z^i,H\}$, allows rewriting
\eqr{\ref{eq:fevolve}} in an explicit conservation law form
\begin{align}
  \pfraca{t}(\mathcal{J}f) + \pfraca{z^i}
  \left(
    \mathcal{J}\dot{z}^i f
  \right) = 0,
  \label{eq:flocal_cons}
\end{align}
where $\mathcal{J}$ is the Jacobian of the phase-space transformation from canonical to
(potentially) non-canonical coordinates, 
and the Liouville theorem on phase-space
incompressibility, $\nabla\cdot(\mathcal{J}\gvec{\alpha}) = 0$, where $\nabla$ is the gradient operator in phase space, has been applied.
Note that if the coordinates are canonical, $\mathcal{J}=1$.

A well-known example of such a system is the equation of vorticity
dynamics in two-dimensional incompressible fluid flow. For this system
the ``phase space'' is $(x,y)$, the ``distribution function''
$f(t,x,y)$ is the vorticity, and the Hamiltonian is simply
$H(x,y)=\phi(x,y)$, where $\phi$ is the potential (or stream function)
determined from $\gcs^2 \phi = -f$. The Poisson bracket operator is
canonical, i.e. $\{f,g\} = f_x g_y - f_y g_x$. The characteristic
velocity is $\dot{\mvec{x}}=\{\mvec{x},H\}=\nabla_\perp\phi\times
\hat{\mvec{z}}$, where $\nabla_\perp=\hat{\mvec{x}}\partial_x +
\hat{\mvec{y}}\partial_y$.


Another example is the Vlasov equation describing the flow of particles in
an electromagnetic field.
For simplicity we focus here on the case of a specified time-independent magnetic field $\mvec{B}(\mvec{x})$ and an electric field from a general electrostatic potential, $\mvec{E}(\mvec{x},t) = - \gcs \phi(\mvec{x},t)$.\footnote{The Hamiltonian structure for the general Vlasov--Maxwell particle-field system is more complicated, see \cite{Marsden:1982,Morrison:1980,Morrison:1984,Morrison:2013}. We have implemented an energy-conserving DG algorithm for the Vlasov equation and the full time-dependent Maxwell's equations using a different technique than described in the present paper, see \cite{Juno:2018}.}
In this case, in the
non-canonical coordinates $z=(\mvec{x},\mvec{v})$, where $\mvec{x}$
are coordinates in the configuration space, and $\mvec{v}$ particle
velocity, the Hamiltonian is given by
\begin{align}
  H(\mvec{x},\mvec{v},t) = \frac{1}{2} m|\mvec{v}|^2 + q \phi(\mvec{x},t) \label{eq:vm-hamil}
\end{align}
where $m$ is the particle mass, $q$ the particle charge, and
$\phi(\mvec{x},t)$ is the electric potential. The corresponding
Poisson bracket is
\begin{align}
  \{f,g\} = 
  \frac{1}{m}
 \big(
    \gcs f \cdot \gvs g - \gcs g \cdot \gvs f
  \big)
  +
  \frac{q\mvec{B}}{m^2} \cdot 
  \gvs f \times \gvs g \label{eq:vm-bracket}
\end{align}
where $\mvec{B}(\mvec{x})$ is the \emph{specified} magnetic field and $\gcs$ and
$\gvs$ are the gradient operators in configuration space and velocity
space, respectively. The Jacobian is a constant,
$\mathcal{J}=m^3$. 
Note that in the case in which there is a non-zero magnetic field, the Poisson
bracket is non-canonical, that is, it can not be written using the commutation operator. This is indicated by the presence of the additional term involving the cross product in \eqr{\ref{eq:vm-bracket}}.
In the $\mvec{B}=0$ limit, the Poisson bracket is canonical, except
for the $1/m$ term. The characteristic velocities are
$\dot{\mvec{x}}=\{\mvec{x},H\}=\mvec{v}$ and
$\dot{\mvec{v}}=\{\mvec{v},H\}=q(-\gcs\phi+\mvec{v}\times\mvec{B})/m$,
as expected.

Unlike Navier--Stokes equations, where there are separate equations for
the conservation of mass, momentum, and energy,
so that finite volume or other 
numerical methods that solve the equations in conservation-law form can automatically ensure those
properties, the conservation properties for Hamiltonian problems in
phase space are indirect. 
Part of the energy conservation requires
that the discrete scheme satisfy one of the quadratic invariants of
the Poisson bracket, i.e. $\int H \{f,H\}\thinspace d\mvec{z} = 0$, with
the integration taken over the complete phase space. Additionally, the
total energy conservation (particles plus field) usually requires use
of the corresponding field equation, and hence also imposes additional
constraints on its solution.

Historically, many of the widely-used numerical methods for simulating
Hamiltonian systems, especially the Vlasov equation, are based on a
particle-in-cell (PIC)
approach \cite{birdsallbook, HockneyEastwood1989}. PIC comprises a class of Lagrangian schemes that introduce macro-particles
that move in phase space with the characteristic velocity $\gvec{\alpha}$. 
(These have similarities to point-vortex and more general vortex methods in hydrodynamics.)
PIC and associated DSMC (Direct Simulation Monte Carlo) methods can be used to simulate very complex
geometries and can incorporate a wide variety of physical processes,
including collisions, chemical reactions, atomic physics, and surface
effects such as sputtering and field-induced emission.
An advantage of such methods is that they have the
ability to provide, in some cases, reasonable results using just a few
(sometimes on the order tens to thousands) macro-particles per
configuration-space cell. 
However, they can have difficulties with
noise for some problems, resulting in the need to use more particles per cell to improve the signal-to-noise ratio. Birdsall and Landgon \cite{birdsallbook}
Chapter 12, and Nevins \emph{et al.} \cite{Nevins:2005fk}, for example, give a
discussion of noise in PIC codes. Note that PIC methods still
require gridding of configuration space (but not of the full
phase space) to account for 
the solution of field equations.

%
Another approach to solving Hamiltonian equations is to directly
discretize the equations using a continuum Eulerian scheme. Unlike
PIC/DSMC algorithms, until recently there had been less work on such
approaches in many areas.
%
%
One reason for this is the increased complexity of developing efficient
continuum schemes due to the high dimensionality of the phase space: in
addition to configuration space, the velocity space also needs to be
discretized. This makes efficient high-order schemes attractive, and the discontinuous Galerkin (DG) family of schemes \cite{Cockburn:2001vr, Cockburn:1998vt} are one such method that has been of growing interest for some computational fluid-dynamics applications \cite{Vincent:2011ez}. However, for Hamiltonian systems the discretization scheme must also preserve conservation laws, and while PIC and some continuum algorithms can conserve energy and momentum exactly, special care must be taken in a DG scheme. As we will describe in this paper, the flexibility of basis and test functions in DG allows for the scheme to be designed such that conserved quantities of the continuous  Hamiltonian system, such as
total (fluid plus field) energy or momentum, can be conserved exactly (or at least to machine precision).


Our algorithms are an extension of the discontinuous Galerkin scheme
presented by Liu and Shu \cite{Liu:2000ee} for the incompressible 2D
Navier--Stokes equations. In the Liu and Shu algorithm a discontinuous
basis set is used to discretize the vorticity equation, while a
continuous basis set is used to discretize the Poisson equation. It
was shown that discrete energy is conserved exactly by the
\emph{spatial scheme} for 2D incompressible flow if basis functions
for potential are a continuous subset of the basis functions for the
vorticity, irrespective of the numerical fluxes selected for the
vorticity equation. Further, they show that the spatial scheme also
conserves the enstrophy (the $L_2$ norm of the vorticity) exactly if a
central flux is used, decaying it for upwind fluxes.  Thus the Liu--Shu
algorithm is a kind of extension of the famous Arakawa
finite difference method \cite{Arakawa1966119,Lilly:1997} to
higher-order DG, and furthermore preserves energy conservation even if
limiters are used on the fluxes.

To our knowledge the connection between the Liu and Shu scheme and
general Hamiltonian systems
has not yet been made. 
In Bernsen and Bokhove \cite{Bernsen2006719} the structure of the
algorithm was further elaborated and extended for 2D geostrophic flow
problems and more sophisticated boundary conditions. 
Einkemmer and Wiesenberger \cite{Einkemmer:2014} present a DG algorithm for 2D incompressible fluids with a Local DG (LDG) algorithm for determining the potential (unlike the Liu and Shu DG algorithm which uses a continuous finite element method for the potential).
However, neither these
papers nor the original Liu and Shu paper points out that these methods
can work for general Hamiltonian problems that can be expressed in terms
of a Poisson bracket.  In fact, a later paper by Ayuso, Carrillo, and
Shu \cite{ayuso2011} does pioneering applications of DG to the Vlasov--Poisson equations but uses a different version of DG than in \cite{Liu:2000ee}, and state that they believe it is the first algorithm to show energy conservation for the Vlasov--Poisson equations.
One advantage of the method we will
use here is that it requires only a single $d$-dimensional Poisson
equation solve each step, while the Ayuso \emph{et al.}\ algorithm requires
$2d$ solutions of $d$-dimensional Poisson equations \cite{Ayuso-de-Dios:2012}.  
Another advantage is that the Ayuso \emph{et al.} algorithm was proved to work only
for DG basis functions in polynomial spaces of degree $p \ge 2$, while
the algorithm we use here works even for $p=1$.


%
%

One area where continuum algorithms (apart from DG) have  been used extensively has been in the development of codes for kinetic plasma systems. 
Various continuum methods have been used in Vlasov/hybrid codes \cite{Cheng:1976, Valentini:2007, Rossmanith:2011, Coulette:2014, vonAlfthan:2014, Kormann:2019} (see also the review by Palmroth \cite{Palmroth:2018}.) Continuum codes have also been used to solve the gyrokinetic equations (a low-frequency asymptotic expansion of the Vlasov-Maxwell equations) for studying plasma turbulence in tokamak fusion devices and certain astrophysical applications \cite{Dorland:2000etg, Jenko:2001, Candy:2003, Nunami:2012, Idomura:2008hx, Peeters:2009, NUMATA:2010, Grandgirard:2016, Candy:2016, Dorr:2018cogent, Mandell:2018}.
These codes have become quite sophisticated and use a combination of various spectral, finite-difference, finite-volume, and semi-Lagrangian methods, but not the DG algorithms we explore here.
%
While the tokamak codes have been quite successful in comparisons with experiments in the core region of tokamaks, most of these codes assumed small-amplitude fluctuations and were optimized for that core region, and so have difficulties when applied to the edge region of fusion devices, where certain underlying approximations are no longer valid.
In this paper we demonstrate a DG algorithm that we believe has certain advantages when applied to the edge region of fusion devices (and other problems), including the ability to conserve energy well for Hamiltonian systems even when using upwind fluxes that aid robustness.

There is now a growing body of interesting work developing various versions of DG and applying them to various kinetic problems.  
Some versions of DG can conserve momentum exactly (but not energy) when applied to Vlasov--Poisson equations \cite{Heath:2012, Cheng:2013bs}.  
We have not yet found a DG algorithm that can conserve energy and momentum simultaneously.  
It would be interesting to compare momentum and energy conserving schemes in a future work.  
Energy-conserving versions of DG have been developed and applied to Vlasov--Maxwell equations \cite{Cheng:2014jcp, Cheng:2014b, Juno:2018}, but the issues are somewhat different because Maxwell's equations are hyperbolic while the Vlasov--Poisson, electrostatic gyrokinetic, and other Hamiltonian systems we consider here have fields determined by elliptic Poisson-type equations.
Various types of sparse-grid versions of DG are being applied to Vlasov--Maxwell systems \cite{Kormann:2015-tensor-train, Guo:2016, Juno:2018, Tao:2019}.
DG is being applied to non-plasma kinetic transport problems, including Boltzmann--Poisson equations for modelling nanoscale semiconductors \cite{Cheng:2009hm, Cheng:2011, Morales-Escalante:2017}.


The rest of this paper is organized as follows. We first describe the
mixed DG/CG algorithm in the context of Vlasov--Poisson equation,
proving energy conservation once particular basis functions are
selected for the field solver. We also show that for the
Vlasov--Poisson equation, the momentum is not conserved, and the
non-conservation is traced to the continuity properties of the
electric field. Although our energy-conserving algorithm does not
conserve momentum exactly, as we show below, the errors in momentum
converge rapidly as the spatial grid is refined (and is independent of
the velocity grid, so a coarse grid in velocity can still be used).
Several example applications
of the scheme to Vlasov--Poisson equations and incompressible Euler
equations are then shown.

\section{The Basic Algorithm}
\label{sec:algo}

In this section we describe the basic algorithm for the solution of
the Hamiltonian evolution equation, \eqr{\ref{eq:fevolve}}, written in
the conservation-law form \eqr{\ref{eq:flocal_cons}}. For simplicity,
we assume that $\mathcal{J}$ is a constant, valid for both the
incompressible Euler equation ($\mathcal{J}=1)$ as well as the
Vlasov--Maxwell/Poisson equations ($\mathcal{J}=m^3)$. Hence, the
equation we wish to solve can be written as a non-linear advection
equation in phase space,
\begin{align}
  \pfrac{f}{t} +
  \nabla\cdot(\gvec{\alpha} f) 
  = 0,
\end{align}
where $\nabla$ is the phase-space gradient operator. The Liouville theorem
on incompressibility gives $\nabla\cdot\gvec{\alpha} = 0$. Also, for
any smooth function $g$, $\{g,H\} = \nabla g\cdot\gvec{\alpha}$ and
$\{g,g\}=0$, hence implying that $\{H,H\}=\nabla H \cdot \gvec{\alpha}
= 0$.

To discretize this equation, we introduce a phase-space mesh
$\mathcal{T}$ with cells $K_j \in \mathcal{T}$, $j=1,\ldots,N$ and
introduce the following piecewise-polynomial approximation space for
the distribution function $f(t,\mvec{z})$
\begin{align}
  \mathcal{V}_h^p = \{ v : v|_K \in \mvec{P}^p, \forall K \in
  \mathcal{T} \},
\end{align}
where $\mvec{P}^p$ is some space of polynomials. 
[That is, $v(z)$ are polynomial functions of $\mvec{z}$ in each cell, and $\mvec{P}^p$ is the space of the linear combination of some set of multi-variate polynomials (the choice of this set is left arbitrary at this point).]
To approximate the
Hamiltonian, on the other hand, we introduce the space
\begin{align}
  \mathcal{W}^p_{0,h} = \mathcal{V}_h^p \cap C_0(\mvec{Z}),
\end{align}
where $\mvec{Z}$ is the phase-space domain and $C_0(\mvec{Z})$ is the set of continuous functions. Essentially, we allow the
distribution function to be discontinuous, while requiring that the
Hamiltonian is in the continuous subset of the space used for the
distribution function. 


The standard formulation of DG then states the problem as finding the time evolution of the discrete approximation $f_h\in \mathcal{V}_h^p$ such that, for all cells $K_j\in \mathcal{T}$ , the error is zero in the weak sense (when projected onto the solution space $\mathcal{V}_h^p$):
\begin{align}
  \int_{K_j} w \left( \pfrac{f_h}{t}  +  \nabla\cdot(\gvec{\alpha}_h f_h) \right)
  \thinspace d\mvec{z} = 0
 \end{align}
for all test functions  $w\in \mathcal{V}_h^p$.  The subscript $h$ indicates the discrete solution.
Integrating by parts within a cell to move derivatives off of $f_h$ (which can be discontinuous across boundaries) onto the test functions $w$ leads to:
\begin{align}
  \int_{K_j} w\pfrac{f_h}{t} \thinspace d\mvec{z} + 
  \oint_{\partial K_j}w^-
  \mvec{n}\cdot\gvec{\alpha}_h\hat{F} \thinspace dS 
  - \int_{K_j}
  \nabla w \cdot \gvec{\alpha}_h
  f_h \thinspace d\mvec{z} = 0. \label{eq:dis-weak-form}
\end{align}
Here, $\mvec{n}$ is an
outward unit vector on the surface of the cell $K_j$, and
the notation $w^-$ ($w^+$) indicates that the
function is evaluated just inside (outside) of the
surface $\partial K_j$.
In this \emph{discrete weak-form}, $f_h^-$ on the surface has been replaced with 
$\hat{F} = \hat{F}(f_h^-, f_h^+)$, a
numerical flux function. The use of a numerical flux function that is common to both sides of a cell face means that the flux of particles out of one cell through a particular face is is identical to the flux into the adjacent cell through that face, and thus ensures particle conservation holds.  (DG borrows this numerical flux function approach from finite volume methods in computational fluid dynamics.)

Note that in writing \eqr{\ref{eq:dis-weak-form}} we have used the
fact that $K_j\in\mathcal{T}$, $\mvec{n}\cdot\gvec{\alpha}^+_h =
\mvec{n}\cdot\gvec{\alpha}^-_h$ on the surface of a cell,
which follows from the restriction that
the Hamiltonian lies in a continuous subset of the basis
for the distribution function.
That is, the following Lemma holds:
\begin{lemma}\label{lem:norm-alpha}
  The component of the characteristic velocity normal to a face of a
  phase-space cell is continuous.
\end{lemma} 
\begin{proof}
  Observe that for a general Hamiltonian
  system, the Poisson bracket operator is defined as [see
  Eq.\thinspace(2.31) in Cary and Brizard\cite{Cary:2009dl}]
  \begin{align}
    \{f,g\} = \pfrac{f}{z^i}\Pi^{ij}\pfrac{g}{z^j},
  \end{align}
  where $\Pi^{ij}$ is the anti-symmetric \emph{Poisson tensor}. The
  characteristic velocity, $\dot{z}^i=\{z^i,H\}$, can then be written
  as $\dot{z}^i=\Pi^{ij}\partial H/\partial z^j$. Let $n_i$ be a unit
  vector normal to a cell surface. We have
  \begin{align}
    n_i \dot{z}^i = n_i \Pi^{ij}\pfrac{H}{z^j} = \tau^j \pfrac{H}{z^j} =
    \gvec{\tau}\cdot\nabla H,
  \end{align}
  where $\tau^j \equiv n_i \Pi^{ij}$. Hence, $\gvec{\tau}\cdot\mvec{n}
  = n_i\Pi^{ij} n_j = 0$, as $\Pi^{ij}$ is anti-symmetric, showing
  that the vector $\gvec{\tau}$ is orthogonal to $\mvec{n}$, and thus
  tangent to the cell surface. Hence, as the Hamiltonian is
  continuous, the tangential component of its gradient (the normal
  component of the characteristic velocity) is also continuous.
\end{proof}
\begin{remark}
In some non-canonical cases (such as electromagnetic gyrokinetics in non-orthogonal field-aligned coordinates), the Poisson tensor itself is not continuous across cell interfaces. In these cases, a numerical flux function should be used for the entire quantity $\mvec{n}\cdot\gvec{\alpha}_h F$ in \eqr{\ref{eq:dis-weak-form}}. This allows the scheme to conserve energy even when Lemma \ref{lem:norm-alpha} does not hold. However for simplicity, we will assume in the following that Lemma \ref{lem:norm-alpha} does indeed hold. 
\end{remark}

The polynomial space $\mvec{P}^p$ in each cell can be
spanned by either the Lagrange tensor basis functions or the
\emph{serendipity} basis functions\cite{Arnold:2011eu}, or some other suitable functions. The Serendipity basis set
have the advantage of using fewer basis functions while giving the same formal
convergence order (though being less accurate) as the Lagrange tensor
basis. Other choices are also possible, for example letting
$\mvec{P}^p$ be the set of all polynomials of order at-most $p$. This
leads to even fewer basis functions, which may be advantageous in higher
phase-space dimensions.

For the numerical flux we use
\begin{align}
  \mvec{n}\cdot\gvec{\alpha}
  \hat{F}(f_h^-,f_h^+)
  =
  \frac{1}{2}\mvec{n}\cdot\gvec{\alpha}
  (f_h^+ + f_h^-)
  -
  \frac{c}{2}(f_h^+-f_h^-),
  \label{eq:num-flux}
\end{align}
where we use either a \emph{central flux} $c = 0$ or an
\emph{upwind flux} $c =|\mvec{n}\cdot\gvec{\alpha}_h|$. The surface
and volume integrals in \eqr{\ref{eq:dis-weak-form}} can be replaced
by Gaussian quadrature of an appropriate order to ensure that the
discrete integrals are performed \emph{exactly}. Under-integration can
lead to subtle problems with stability and energy
conservation, as we discussed in \cite{Juno:2018}.
Finally, we remark that a strong-stability-preserving
(SSP), third-order Runge--Kutta scheme is used here to advance the solution
in time, but other time integration algorithms can be used as well.

The discretization of the Hamiltonian depends on the physical system
under consideration. For the Vlasov--Poisson equations the discrete
form of each species Hamiltonian, \eqr{\ref{eq:vm-hamil}}, can be
written as
\begin{align}
  H_h(\mvec{x},\mvec{v},t) = \frac{1}{2} m v_h^2 + q
  \phi_h(\mvec{x},t),
  \label{eq:vp-hamil}
\end{align}
where ${v}^2_h$ is the projection of $v^2$ on the space
$\mathcal{W}^p_{0,h}$ (that is, the kinetic-energy term must be
projected onto the continuous space). 
For example, for piecewise-linear basis functions in 1D, $v^2_h$ is a continuous piecewise linear approximation to $v^2$ (this allows energy conservation even if $v^2$ is not in the basis set), while if piecewise parabolic basis functions are used, then $v^2_h = v^2$.  
The electrostatic potential is
determined from the Poisson equation
\begin{align}
  \gcs^2 \phi(\mvec{x},t) 
  = -\frac{1}{\epsilon_0}
  \sum_s q \int_{-\infty}^{\infty} f(t,\mvec{x},\mvec{v})
  \thinspace d\mvec{v}
  = -\frac{\varrho_c}{\epsilon_0}, \label{eq:vp-poisson}
\end{align}
where $\epsilon_0$ is the permittivity of free space, $\varrho_c$ is
the total charge density, and the sum extends over all species in the
plasma. (The species subscript on the charge, mass, and distribution
function are dropped.) To discretize the Poisson equation we use the
solution space
\begin{align}
  \mathcal{X}^p_{0,h} = \mathcal{W}^p_{0,h} \backslash \Omega,
\end{align}
i.e, the restriction of the continuous set $\mathcal{W}^p_{0,h}$ on
the configuration space $\Omega$.  
The standard continuous Galerkin / finite element method for Poisson-type elliptic equations is to find $\phi_h\in \mathcal{X}^p_{0,h}$ such that
\begin{align}
  \oint_{\partial\Omega} \psi\gcs \phi_h \cdot \mvec{n} dS
  -\int_\Omega \nabla\psi \cdot \gcs \phi_h\thinspace d\mvec{x}
  =
  -\frac{1}{\epsilon_0}\int_\Omega \psi\varrho_{ch} \thinspace d\mvec{x}
  \label{eq:dis-poisson}
\end{align}
for all test functions $\psi\in \mathcal{X}^p_{0,h}$. Here, the integration is
performed over the complete configuration space, and $\mvec{n}$ is a
unit outward normal to the configuration-space boundary.

\section{Conservation and stability properties of the scheme}

Hamiltonian-evolution equations satisfy several conservation
laws, which follow from the properties of the Poisson bracket
operator. In particular, the identities
\begin{align}
  \int f\{f,H\}\thinspace d\mvec{z} = \int H\{f,H\}\thinspace d\mvec{z} = 0,
\end{align}
where the integration is taken over the entire phase space, lead to the
conservation of the \emph{quadratic invariants}: the $L_2$ norm of the
distribution function (called \emph{enstrophy} for the incompressible
Euler equations, and related to the entropy for Vlasov--Poisson
equations) and the energy, respectively. For accuracy and physical
robustness, particularly for long time scales, it is important that the 
numerical scheme conserve the energy.  For numerical stability it is 
important that the $L_2$ norm be
conserved or decay monotonically. In addition, systems like
Vlasov--Maxwell/Poisson equations also satisfy other conservation laws
like total (particle plus fields) momentum. The total
number of ``particles'' ($\int f\thinspace d\mvec{z}$) is also
conserved (this is equivalent to the circulation or integrated vorticity in 
the 2D incompressible hydrodynamics case). 
In this section we examine the conservation and stability
properties of the scheme presented in the previous section. Appropriate
boundary conditions (or behavior at infinity) is assumed to eliminate global
surface terms in the following discussion.

\begin{proposition}
  The total number of particles is conserved exactly.
\end{proposition}
\begin{proof}
  Conservation of total number of particles follows immediately on
  selecting $w=1$ in the discrete weak-form,
  \eqr{\ref{eq:dis-weak-form}}, and summing over all cells leading to
  \begin{align}
    \pfraca{t} \sum_{K_j\in\mathcal{T}}\int_{K_j} f_h d\mvec{z} = 0.
  \end{align}
\end{proof}

\begin{proposition}
  The spatial scheme conserves total energy exactly.
\end{proposition}
\begin{proof}
  To prove energy conservation, select $w=H_h$ in
  \eqr{\ref{eq:dis-weak-form}}, to write
  \begin{align}
    \int_{K_j} H_h\pfrac{f_h}{t}d\mvec{z} + 
    \oint_{\partial K_j}H_h^-
    \mvec{n}\cdot\gvec{\alpha}_h\hat{F}dS 
    - \int_{K_j}
    \nabla H_h \cdot \gvec{\alpha}_h
    f_h d\mvec{z} = 0
  \end{align}
  As $\{H_h,H_h\}=\nabla H_h\cdot\gvec{\alpha}_h=0$, the last term
  vanishes. Physically, this is because the flow $\gvec{\alpha_h}$ is
  along contours of constant energy in phase space. On summing over
  all cells, the contribution from the surface integral term drops as
  the Hamiltonian and the numerical flux are both continuous and $\mvec{n}\cdot\gvec{\alpha}_h$
  differs only in sign for the two cells sharing a face. Thus we get
  \begin{align}
    \sum_{K_j\in\mathcal{T}} \int_{K_j} H_h\pfrac{f_h}{t}d\mvec{z} = 0.
    \label{eq:dis-ht-cons}
  \end{align}
  Note that this result holds for a general Hamiltonian system, and is
  independent of the numerical flux function selected.

  The above proves conservation of energy
  for the case where $H$ is static. For
  the more general time-dependent case, one must consider the field
  equation that relates the potential (and thus $H$) to $f$.  For
  example, to prove total (particle plus field) energy conservation
  for the Vlasov--Poisson equations, one must consider two pieces: the particle kinetic energy and the field energy. That
  is, we write $W=W_k + W_E$, where 
  the particle kinetic energy is $W_k = \sum_s \int \frac{1}{2}mv^2 f_h d\mvec{z}$, and the field energy is $W_E = \frac{\epsilon_0}{2}\int |\gcs \phi_h|^2$. Note that we can also write the total  energy as $W=W_H-W_E$, where $W_H = \sum_s \int H_h f_h$ is the total (kinetic plus twice-potential) particle energy.  
  [The term $\sum_s q \int \phi_h f_h d\mvec{z}$ in $W_H$ is the potential energy when the potential is externally imposed, but for self-consistent electrostatic interactions it double counts the potential energy from every pair of particle interactions. This results in the need to subract the field energy $W_E$ from $W_H$ to get the total energy in the system.]

  To show that the discrete scheme indeed conserves total energy for
  the Vlasov--Poisson equation, substitute \eqr{\ref{eq:vp-hamil}} in
  \eqr{\ref{eq:dis-ht-cons}} and sum over all species in the plasma
  \begin{align}
    \sum_s \sum_{K_j\in\mathcal{T}} \int_{K_j} 
    \left(
      \frac{1}{2} m v_h^2 + q
  \phi_h(\mvec{x},t)
    \right)
    \pfrac{f_h}{t}d\mvec{z} = 0.    
  \end{align}
  Recognizing the first term as the total particle kinetic energy, this can be written as 
  \begin{align}
    \frac{\partial W_k}{\partial t} + \sum_{\Omega_j\in\mathcal{T}_\mvec{x}} \int_{\Omega_j} 
      \phi_h(\mvec{x},t)\pfrac{\varrho_{ch}}{t}
    \thinspace d\mvec{x} = 0,
    \label{eq:tot-energy-t1}
  \end{align}
%
  where $\mathcal{T}_\mvec{x}$ is the mesh in configuration
  space. Taking the time derivative of the discretized field equation \eqr{\ref{eq:dis-poisson}} and
  dropping surface terms (which vanish due to assumed boundary
  conditions) we get
  \begin{align}
    \int_\Omega \gcs \psi \cdot \gcs \pfrac{\phi_h}{t} \thinspace d\mvec{x}
    =
    \frac{1}{\epsilon_0}\int_\Omega \psi\pfrac{\varrho_{ch}}{t} \thinspace d\mvec{x}.
  \end{align}
  Selecting $\psi=\phi_h$ in this and using the resulting expression
  in \eqr{\ref{eq:tot-energy-t1}} we get
  \begin{align}
    \pfrac{W_k}{t}
    +
    \pfrac{W_E}{t} = 0
    \label{eq:tot-energy}
  \end{align}
  which shows that the spatial scheme conserves the total energy for the
  Vlasov-Poisson equations exactly.
\end{proof}
We remark that the exact form of the energy conservation proof depends
on the particular physical system under consideration and the above
procedure would need slight modification. In general, one needs to
show that $dW/dt = 0$, where $W = \int f_h H_h d\mvec{z}-W_E$, with $W_E$ containing possible field terms, and then use
the field equations to compute the time-derivative of the Hamiltonian
that appears in this to construct a total-energy. For example, for the long-wavelength
gyrokinetic equations the proofs that our scheme conserves the
discrete total energy are given in \cite{mandell-2019, Shi2017thesis,
  Shi2015}.

\begin{proposition}\label{prop:l2}
  The spatial scheme exactly conserves the $L_2$ norm of the
  distribution function when using a central flux, while the distribution-function $L_2$ norm monotonically decays when using an upwind flux.
\end{proposition}
\begin{proof}
  The proof follows Liu and Shu \cite{Liu:2000ee}. Use $w=f_h$ in
  \eqr{\ref{eq:dis-weak-form}}, to write
  \begin{align}
    \pfraca{t} \int_{K_j} \frac{1}{2}f_h^2 \thinspace d\mvec{z} + 
    \oint_{\partial K_j}f_h^-
    \mvec{n}\cdot\gvec{\alpha}_h\hat{F} \thinspace dS 
    - \int_{K_j}
    \nabla f_h \cdot \gvec{\alpha}_h f_h \thinspace d\mvec{z} = 0.
  \end{align}
  From incompressibility, we can write $\nabla f_h \cdot
  \gvec{\alpha}_h f_h = \nabla\cdot (\gvec{\alpha}_h f_h^2/2)$, and
  performing an integration by parts we get
  \begin{align}
    \pfraca{t} \int_{K_j} \frac{1}{2}f_h^2 \thinspace d\mvec{z} + 
    \oint_{\partial K_j}
    \mvec{n}\cdot\gvec{\alpha}_h f_h^-
    \big(
      \hat{F} - \frac{1}{2} f_h^-
    \big) \thinspace dS 
    = 0.
  \end{align}
  If a central flux $\hat{F}=(f_h^+ + f_h^-)/2$ is used, then the
  integrand of the second term becomes $\mvec{n}\cdot\gvec{\alpha}_h
  f_h^- f_h^+/2$, which has opposite signs for the two cells sharing a
  face. Hence, summing over all cells gives the conservation law
  \begin{align}
    \pfraca{t} 
    \sum_{K_j\in\mathcal{T}} \int_{K_j} \frac{1}{2} f_h^2 \thinspace d\mvec{z} = 0.
  \end{align}
  For an upwind flux [$c=|\mvec{n}\cdot\gvec{\alpha}|$ in
  \eqr{\ref{eq:num-flux}}], the integrand in the second term becomes
  \begin{equation}
  \mvec{n}\cdot\gvec{\alpha}_h f_h^- f_h^+/2 -
  |\mvec{n}\cdot\gvec{\alpha}_h|[(f_h^+)^2-(f_h^-)^2]/4 +
  |\mvec{n}\cdot\gvec{\alpha}_h|(f_h^+-f_h^-)^2/4.
  \end{equation}
  Summing over all cells
  and dropping terms that vanish on summation, we get
  \begin{align}
    \pfraca{t} 
    \sum_{K_j\in\mathcal{T}} \int_{K_j} \frac{1}{2} f_h^2 \thinspace
    d\mvec{z} 
    + \frac{1}{4} \sum_{K_j\in \mathcal{T}} \oint_{\partial K_j}
    |\mvec{n}\cdot\gvec{\alpha}| (f_h^+-f_h^-)^2 dS
    = 0.
  \end{align}
  As the second integral is always positive, this shows that the $L_2$
  norm decays monotonically when using an upwind flux. This
  proposition shows that the scheme is stable in the $L_2$ norm of the
  distribution function.
\end{proof}

It should be remarked that the proof of energy conservation and $L_2$
conservation with a central flux (or monotonic decay of $L_2$ with an upwind flux)
critically relies on the fact that the Hamiltonian is continuous, from
which it follows that the normal component of the characteristic
velocity is continuous, as shown in
Lemma\thinspace\ref{lem:norm-alpha}. This, in combination with the
fact that the fields appearing in the Hamiltonian lie in a continuous
subspace of the space containing the distribution function, leads to
energy conservation.

Another remark is that even though the spatial scheme conserves total
energy, the fully discrete scheme (including time discretization) generally does
not. That is, unless a time-reversible scheme is used to advance time,
the total energy will be conserved only to the order of the
time-discretization scheme. 
However, the conservation errors, even for a non-reversible scheme, will
be \emph{independent} of the phase-space discretization, and can be
reduced by taking a smaller time step, if desired.

\begin{proposition}\label{prop:entropy}
  If the discrete distribution function $f_h$ remains positive
  definite, then the discrete scheme grows the discrete entropy
  monotonically,
  \begin{align}
    \frac{d}{dt}\int_{K_j} -f_h \ln(f_h) \thinspace d\mvec{z} \ge 0.
  \end{align}
\end{proposition}
\begin{proof}
  We have the bound $\ln(f_h) \le f_h-1$ as long as
  $f_h>0$. Multiplying by $-f_h$ gives us the inequality
  \begin{align}
    -f_h \ln(f_h) \ge -f_h^2 + f_h.
  \end{align}
  The left-hand side is just the discrete entropy for the Vlasov
  equation. Integrating over a phase-space cell, summing over cells,
  and taking the time derivative of both sides gives us an expression
  for the time evolution of the discrete entropy in our scheme
  \begin{align}
    \sum_j \frac{d}{dt} \int_{K_j} -f_h \ln(f_h) \thinspace d\mvec{z}
    \ge \sum_j \frac{d}{dt} \int_{K_j} \left( -f_h^2 + f_h \right) \thinspace d\mvec{z}.
  \end{align}
  The second term on the right-hand side vanishes as number density is
  conserved. Also, the discrete entropy either
  remains constant, or increases monotonically, as we have previously shown that $f_h^2$ remains constant or
  decays monotonically.
\end{proof}
Note that the scheme presented above does not ensure that $f_h>0$ by construction. Maintaining positivity and conservation simultaneously for Hamiltonian systems can be subtle as the conservation is indirect, involving integration by parts and exchange of energy between particles and fields. We will present a novel algorithm for positivity that continues to maintain total conservation for kinetic systems in a future publication.

In addition to particles and total energy, the Vlasov--Poisson
equations also conserve total momentum. Multiplying the Vlasov
equation by $m \mvec{v}$, integrating over all phase space and summing
over plasma species we get
\begin{align}
  \pfraca{t}  \sum_s \int_\Omega\int_{-\infty}^\infty 
  m\mvec{v} f \thinspace d^3\mvec{v}\thinspace d\mvec{x}
  +
  \int_\Omega ( \gcs \phi ) \varrho_c \thinspace d\mvec{x}
  = 0. \label{eq:ex-mom-t1}
\end{align}
Using the Poisson equation to eliminate charge density and integrating by parts leads to
$\int_\Omega ( \gcs \phi ) \rho_c d \mvec{x} = 
- \epsilon_0 \int_\Omega ( \gcs \phi ) \gcs^2 \phi =
-  \epsilon_0 \int_\Omega \gcs |\gcs \phi |^2 / 2 = 0
$
so the total momentum is conserved:
\begin{align}
  \pfraca{t}  \sum_s \int_\Omega\int_{-\infty}^\infty 
  m\mvec{v} f \thinspace d^3\mvec{v}\thinspace d\mvec{x}
  = 0.
  \label{eq:ex-mom}
\end{align}
Periodic boundary conditions were used to eliminate surface terms in
the above expressions. Note that for electrostatic problems the (gradient of) potential does not appear in the momentum, i.e. an
electrostatic field does not carry momentum. For the Vlasov--Maxwell
equations, on the other hand, the \emph{total} momentum contains terms proportional to $\mvec{E}\times\mvec{B}$.

To check if the scheme conserves momentum for the Vlasov--Poisson
equations, we set $w= m \mvec{v}_h$ in the discrete weak-form,
\eqr{\ref{eq:dis-weak-form}}, and sum over all cells and species in
the plasma to get
\begin{align}
  \sum_s \sum_{K_j\in\mathcal{T}} \int_{K_j}
  \left(
    m \mvec{v}_h \pfrac{f_h}{t}
    - 
    m \nabla \mvec{v}_h \cdot \gvec{\alpha}_h
    f_h 
  \right)
  \thinspace d\mvec{z} = 0.
\end{align}
We have used the continuity of $\mvec{v}_h$ to drop the surface term
on summation. Introducing the total particle momentum via
\begin{align}
  \mvec{M}_h(\mvec{x},t) \equiv 
  \sum_s \
  \int_{-\infty}^{\infty}
  m\mvec{v}_h f_h(t,\mvec{x},\mvec{v})\thinspace d^3\mvec{v}
\end{align}
and noticing that $m \nabla \mvec{v}_h \cdot \gvec{\alpha}_h =
m\{\mvec{v}_h,H_h\} = -q \gcs\phi_h$ we get
\begin{align}
  \int_{\Omega}
  \pfrac{\mvec{M}_h}{t}
  \thinspace d\mvec{x}
  +
  \sum_{\Omega_j\in\mathcal{T}_\mvec{x}} \int_{\Omega_j}
    \gcs\phi_h\varrho_{ch}
  \thinspace d\mvec{x} = 0.
  \label{eq:mom-non-con}
\end{align} %
For the scheme to conserve momentum, the second term in the above
expression must vanish. However, the elimination of $\varrho_{ch}$
using the discrete Poisson equation, analogous to the one made in
deriving the exact momentum conservation equation,
(\eqr{\ref{eq:ex-mom}}), cannot be done here as the weak-form
\eqr{\ref{eq:dis-poisson}} does not preserve integration by parts
\emph{locally}. Another way to state this is that one cannot use
$\psi=\gcs\phi_h$ as a test function in \eqr{\ref{eq:dis-poisson}} as
$\gcs\phi_h$ is (in general) discontinuous and hence not in
$\mathcal{W}^p_{0,h}$. Hence, in general, the scheme presented here
will not conserve total momentum exactly.

However, it should be noted that the errors in discrete momentum conservation
depend only on the accuracy of the spatial integral in the second term in \eqr{\ref{eq:mom-non-con}} (which cancels exactly in the continuous limit due to the steps in going from 
\eqr{\ref{eq:ex-mom-t1}} to \eqr{\ref{eq:ex-mom}}).
Thus the scheme will converge towards momentum conservation depending only on the spatial resolution, independent of velocity resolution.

There are other variations of DG that can conserve momentum exactly instead of energy, but we have not found a DG algorithm that can conserve energy and momentum simultaneously.  
This is similar to the situation with the PIC algorithms in Birdsall and Langdon \cite{birdsallbook}, which conserve either momentum or energy but not both.
Conserving total momentum might be more important for some applications, but in principle there could be large positive momentum errors for some particles offset by large negative momentum errors for other particles while still conserving total momentum.  
This seems a less useful constraint on the dynamics than conservation of energy, a positive definite quantity.
This is one of the reasons we chose to conserve energy.

\section{Benchmarks for the Vlasov--Poisson and incompressible Euler
  equations}

In this section we apply the scheme developed in the previous sections
to the Vlasov--Poisson and 2D incompressible Euler equations. For
simplicity, we study the the system in 2D phase-space (1D/1V) and
assume that the ions are stationary. The Hamiltonian for this problem
is $H=m_ev^2/2 -|e|\phi(t,\mvec{x})$, where $m_e$ is the electron mass, and
$-|e|$ the electron charge. The Poisson bracket is
$\{f,g\}= (\gcs f \gvs g - \gvs g \gcs f)/m_e$, which is canonical except for the
$1/m_e$ term. In the first series of tests we use a time-independent,
but spatially varying potential profile. In the second set of tests we
evolve the potential using Poisson's equation,
\eqr{\ref{eq:vp-poisson}}. As in the previous section, we use ``SP1'' to denote the use of first-order serendipity polynomial basis functions and ``SP2'' for second-order serendipity polynomial basis functions.

We should remark that the schemes presented are being used in our
production code {\tt Gkeyll} to solve the far-more-complex gyrokinetic
equations. For example, the scheme was used to study plasma turbulence
in straight and helical field-line
geometries~\cite{Shi2017thesis,Shi2017,Shi2019}, including comparison
with experiments \cite{Bernard:2019be}. Recent work includes extension
to the highly challenging case of the full-$f$ electromagnetic gyrokinetic
system \cite{mandell-2019}. Each of these applications also includes
collision terms modeled using a model Fokker--Planck
operator~\cite{Hakim:2019wu} which is itself discretized using a
momentum- and energy-conserving DG scheme. The problems presented below
test the scheme in simpler settings, but our current physics research
shows that the essential scheme is robust and scalable to very complex
plasma problems. To allow an interested reader to reproduce our results, instructions to obtain \gke\ and the input files used
here are given in the appendix.

\subsection{Conservation properties}

To test momentum and energy conservation a series of simulations with
the Vlasov--Poisson equations are performed. The ions are assumed
to be stationary and only the electron distribution function is evolved. The
electron distribution function is initialized as

\begin{align}
  f(x,v,0) = 
  \left\{
    \begin{array}{l l}
      \{1+\exp\left[-\beta_l(x-x_m)^2\right]\}\thinspace f_m(T_e, v_d) \quad x<x_m \\
      \{1+\exp\left[-\beta_r(x-x_m)^2\right]\}\thinspace f_m(T_e, v_d) \quad x \ge x_m  
    \end{array}
  \right.,
\end{align}
where $\beta_l = 0.75$, $\beta_r = 0.075$, $v_d=1.0$, $T_e=1.0$ and
$x_m=-\pi$. Further, $f_m(T,v_d)$ is a drifting Maxwellian with a
specified temperature and drift velocity:
\begin{align}
  f_m(T,v_d) = \frac{1}{\sqrt{2\pi v_t^2}} \exp[-(v-v_d)^2/2v_t^2],
\end{align}
where $v_t=\sqrt{T/m}$. The initial conditions drive strong asymmetric
flows around $x=x_m$ from the asymmetric number-density profile. Note
that if a symmetric initial profile is used the net initial momentum
in the system is zero and will remain so (to machine precision)
as the solution evolves.

\begin{table}[htbp]
  \caption{Error in momentum conservation for Vlasov--Poisson equations
    with serendipity polynomial order one (SP1), and two (SP2)
    schemes. For each spatial resolution ($N_x$) two velocity space
    resolutions were used ($N_v=32,128$). Simulations are run to
    $t=20$ and errors measured. Momentum errors are insensitive to
    velocity space resolution and converge rapidly with spatial
    resolution and polynomial order.}
  \centering
  \begin{tabular}{ccc|cc}
    \hline
     & \multicolumn{2}{c|}{SP1} & \multicolumn{2}{c}{SP2} \\
    $N_x$ & $N_v$ (32) & $N_v$ (128) & $N_v$ (32) & $N_v$ (128) \\
    \hline
    8 & 1.4052(-3) & 1.3332(-3) & 1.9480(-5)  & 1.9398(-5) \\
    16 & 3.6887(-4) & 3.9308(-4) & 6.9063(-7) &  6.8864(-7) \\
    32 & 6.3612(-5) & 8.5969(-5) & 5.9614(-8) & 5.1174(-8) \\
    64 & 9.0199(-6) & 1.5253(-5) & 2.2169(-9) & 2.2288(-9) \\
    \hline
  \end{tabular}
  \label{tb:mom-err}
\end{table}

Simulations were performed with serendipity basis functions with
polynomial order one (SP1) and two (SP2). For all problems the domain
is $[-2\pi\times 2\pi]\times [-10,10]$. Tests of momentum conservation
are shown in Table~\ref{tb:mom-err}. Even though momentum is not
conserved, the results confirm that the errors in momentum conservation
are insensitive to the velocity space resolution (there is some very weak dependence due to the discrete initial conditions depending on the velocity resolution). In addition, the
errors reduce rapidly with spatial resolution and increasing
polynomial order.

\begin{table}[htbp]
  \caption{Error in energy conservation for Vlasov--Poisson equations
    with serendipity polynomial order one (SP1), and two (SP2)
    schemes. A fixed resolution of $(N_x,N_v)=(16,32)$ is used,
    while the time step is varied by changing the CFL
    number. Simulations are run to $t=20$, at which time the errors are measured. Energy
    errors are solely from the non-reversible nature of the strong
    stability preserving Runge--Kutta order 3 scheme used, and converge
    with the same order as the time-stepping scheme. Errors also
    reduce rapidly with increasing polynomial order.}
  \centering
  \begin{tabular}{ccc|cc}
    \hline
     & \multicolumn{2}{c|}{SP1} & \multicolumn{2}{c}{SP2} \\
    CFL & Error & Order & Error & Order \\
    \hline
    0.3 &  1.4185(-6) &  & 4.1646(-7) & \\
    0.15 & 1.7687(-7) & 3.0 & 5.1978(-8) & 3.0 \\
    0.075 & 2.2078(-8) & 3.0 & 6.4914(-9) & 3.0 \\
    0.0375 & 2.7587(-9) & 3.0 & 8.1295(-10) & 3.0 \\
    \hline
  \end{tabular}
  \label{tb:energy-err}
\end{table}

Energy-conservation tests are performed with the same initial
condition, however with on a fixed grid of $(N_x,N_v)=(16,32)$. As
mentioned, even though the spatial scheme conserves momentum exactly,
use of a SSP Runge--Kutta (RK) scheme (which is not reversible) means
that the time stepping will introduce errors in energy, which scale
as the order of the RK scheme selected. This is clearly seen in
Table~\ref{tb:energy-err} which shows that the energy-conservation
errors reduce as $O(\Delta t^3)$, the same as the order of the RK
scheme selected.

\subsection{Free streaming and recurrence}

In the first test we set $\phi = 0$, which leads to the free streaming
of particles. In this case, the exact solution is $f(x,v,t) =
f(x-vt,v,0)$, \emph{i.e.} at each point in velocity space the initial
distribution advects with a constant speed. However, even though the
distribution function is manifestly undamped, its \emph{moments} are
damped. To see this pick an initial condition a Maxwellian $f_M(x,v) =
1/\sqrt{2\pi v_t^2} \exp(-v^2/2v_t^2) \cos(kx)$, where $v_t$ is the
thermal velocity and $k$ is the wave-number. Then, the exact solution
is $f(x,v,t) = f_M(x-vt,v)=1/\sqrt{2\pi v_t^2}\exp(-v^2/2v_t^2)
\cos\left( k(x-vt) \right)$. The increasingly oscillatory nature of
the $\cos\left[ k(x-vt) \right]$ term results in \emph{phase mixing}
due to which all moments of the distribution function are damped. For
example, the number density is
\begin{align}
  n(x,t) = \int_{-\infty}^\infty f dv = e^{-k^2v_t^2t^2/2} \cos(kx),
\end{align}
which is exponentially damped.

\begin{figure}
  \setkeys{Gin}{width=0.75\linewidth,keepaspectratio}
  \incfig{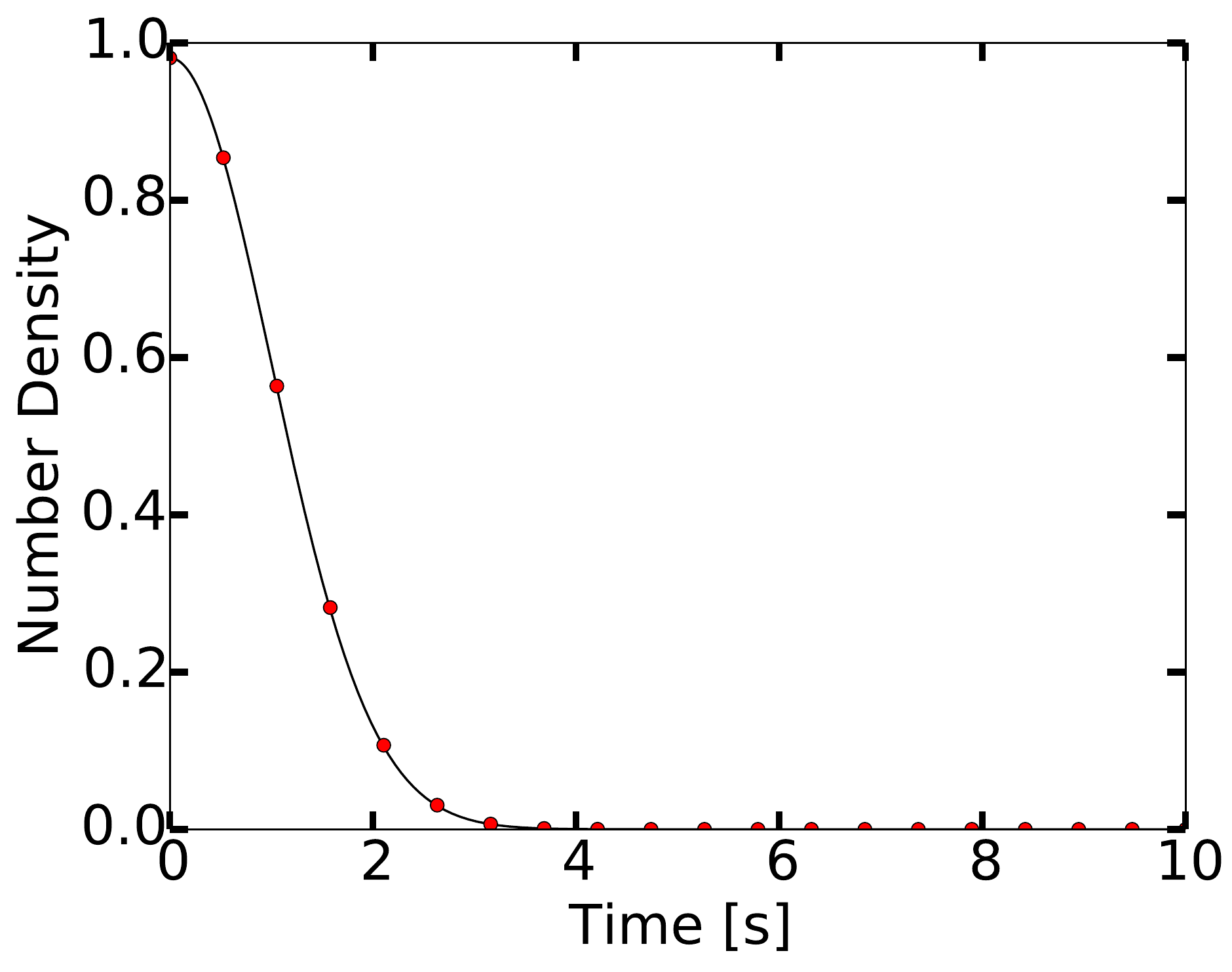}
  \caption{Number density (black) in cell 2 as a function of time
    for a free-streaming problem.
    SP1 basis functions with upwind flux are
    used on a $(N_x,N_v)=(64, 64)$ grid. The red dots show the exact
    solution. At this resolution the numerical results are
    indistinguishable from the exact solution.}
  \label{fig:num-dens-landau-damp}
\end{figure}

To test the ability of the algorithm to model this damping, a
simulation is initialized with a Maxwellian with $v_t=1$ and $k=1$ on
a domain $(x,v) \in [0,2\pi] \times [-6,6]$. The number density is
computed and recorded in a specified cell. The simulation is run on a
$(N_x,N_v)=(64,64)$ grid using SP1 basis functions.
The number density measured in a cell is shown in
Fig.~\ref{fig:num-dens-landau-damp}. At this resolution the numerical
results are indistinguishable from the exact solution.


The use of a discrete grid, however, combined with the lack of
physical (or artificial) dissipation, leads to \emph{recurrence},
i.e. the initial condition will recur after some finite time. This
occurs as the phases of the discrete solution decohere temporarily
leading to collisionless damping, however, after a finite time combine
together again to recreate the initial conditions. A small amount of damping (via collisions or hyper-collisions), however, can easily control this
recurrence even on a coarse grid.

\subsection{Particle motion in specified potential}

In this set of problems the potential $\phi(x)$ is held fixed in time
and the distribution is evolved. These cases correspond to the motion of
test particles in a specified potential. In each case the initial
distribution is assumed to be a uniform Maxwellian
\begin{align}
  f(x,v,0) = \frac{1}{\sqrt{2\pi v_t}} \exp(-v^2/2v_t^2)
\end{align}
with $v_t=1$. If the potential has a well (a minima) then a fraction
of the particles will be trapped and appear as rotating vortices in
the distribution function plots. The bounce period for a particle with
total energy $E$ (which is a constant of motion) in a well can be
computed as
\begin{align}
  T(E) = \sqrt{2m} \int_{x_1(E)}^{x_2(E)} \frac{dx}{\sqrt{E-\phi(x)}},
\end{align}
where $x_1$ and $x_2$ are the roots of the equation $\phi(x)=E$,
\emph{i.e.} the turning points at which the motion of the particle reverses. For finite $x_1$ and $x_2$ the motion is periodic. Note that for
a non-singular distribution (like the Maxwellian) the bounce period
need not be the same for all the particles. In this case an average
period can be computed.

\begin{figure}
  \setkeys{Gin}{width=0.45\linewidth,keepaspectratio}
  \incfig{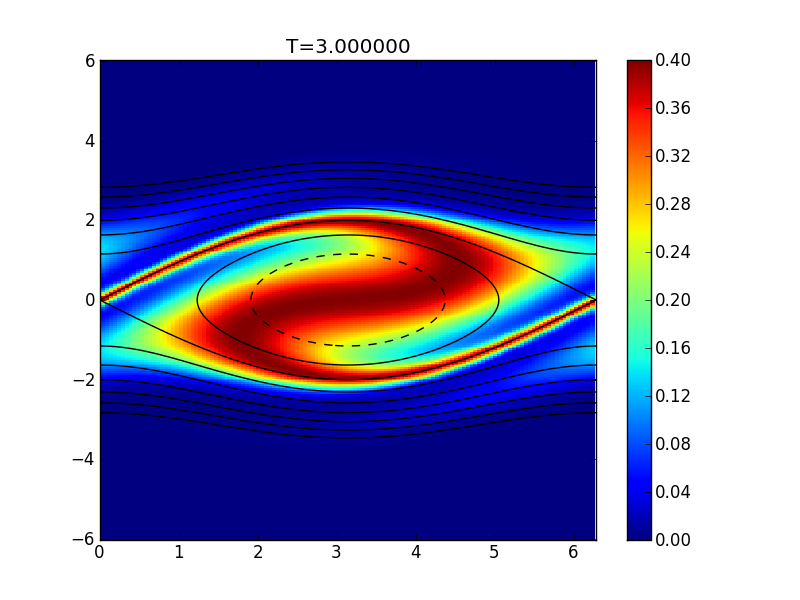}
  \incfig{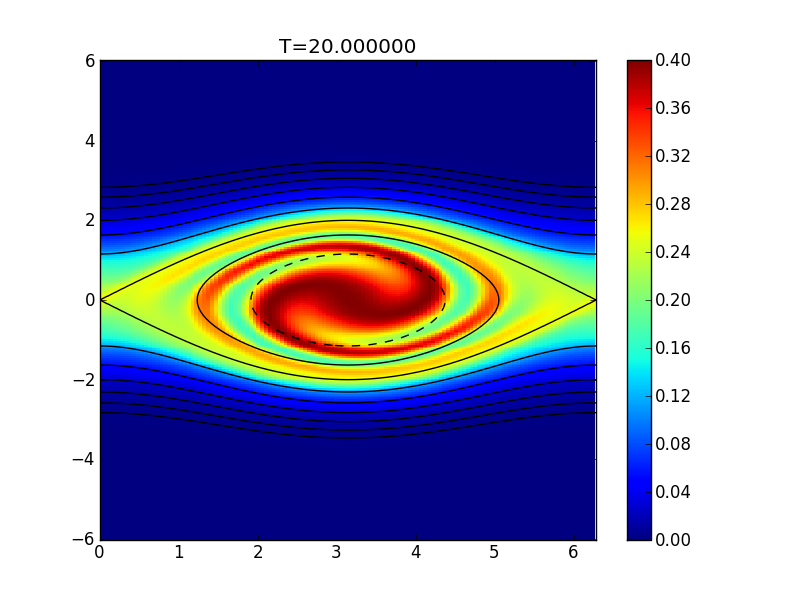}
  \caption{Distribution function at $t=3$ (left) and $t=20$ for flow
    in a $\cos(x)$ potential well. The black lines show contours of
    constant particle energy. A separatrix forms along the
    trapped--passing boundary.}
  \label{fig:cosx-pot}
\end{figure}

For the first test we initialize $\phi(x) = \cos(x)$. Simulations were
run with a SP1 scheme on a $64\times 128$ grid for
$(x,v) \in [0,2\pi] \times [-6,6]$. In this potential the bounce
period of a single particle depends on its initial energy, and the
particles with lower total energy bounce faster. Snapshots of the
distribution function are shown at a $t=3$ and $t=20$ in
Fig.\thinspace\ref{fig:cosx-pot}.

\begin{figure}
  \setkeys{Gin}{width=0.45\linewidth,keepaspectratio}
  \incfig{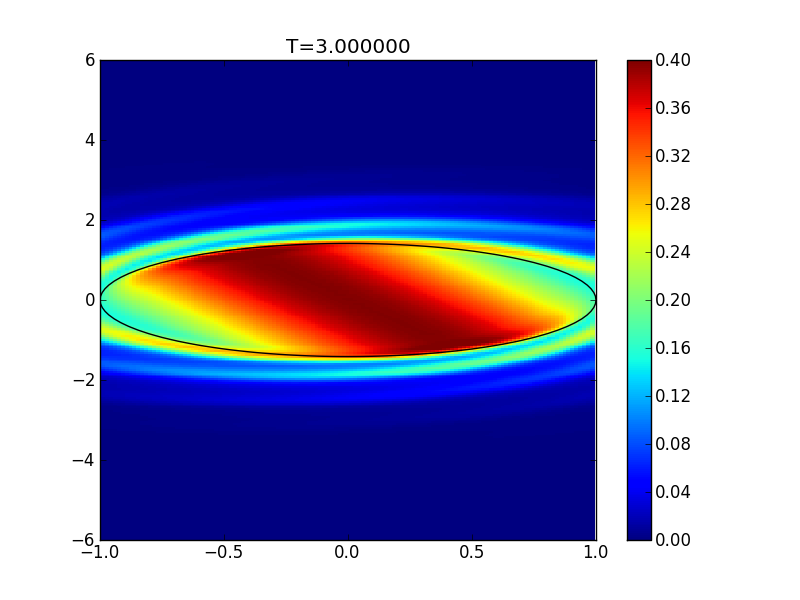}
  \incfig{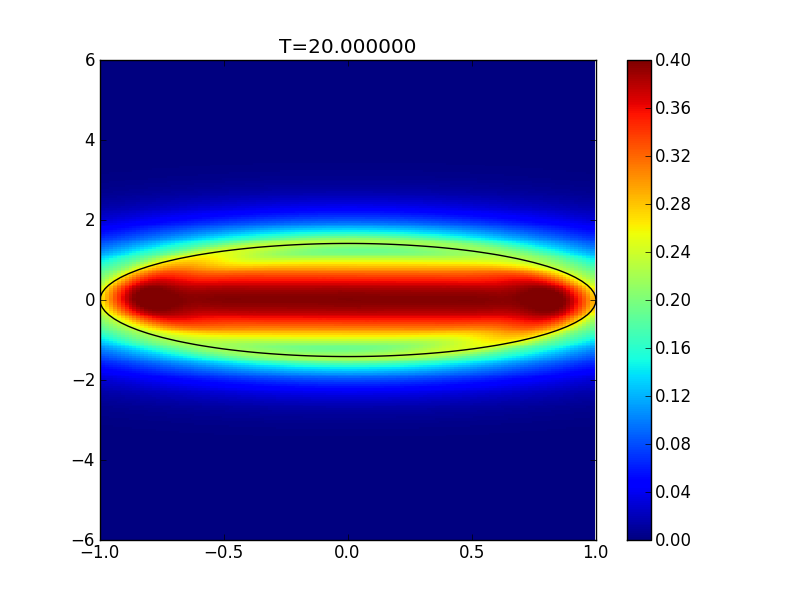}
  \caption{Distribution function at $t=3$ (left) and $t=20$ (right)
    for flow in a potential well. The black line shows the
    trapped--passing energy contour. Due to the quadratic potential $\phi(x) = x^2$,
    all particles inside the trapped region move with the same angular
    velocity in phase space, and the motion appears like
    rigid-body rotation.}
  \label{fig:quad-pot}
\end{figure}

In the second test we set $\phi(x) = x^2$. This quadratic potential
corresponds to simple harmonic motion, and the bounce period of all
trapped particles are the same and can be computed as
$\pi\sqrt{2} \approx 4.443$. Also, as the bounce period for all
trapped particles is the same, these will move ``rigidly'' in
phase-space, \emph{i.e.} the motion along contours of constant energy will
occur with the same frequency. These features are clearly seen in the
snapshots of the distribution function shown at $t=3$ and $t=20$ in
Fig.\thinspace\ref{fig:quad-pot}.

\subsubsection{Landau Damping}

\begin{figure}
  \setkeys{Gin}{width=0.45\linewidth,keepaspectratio}
  \incfig{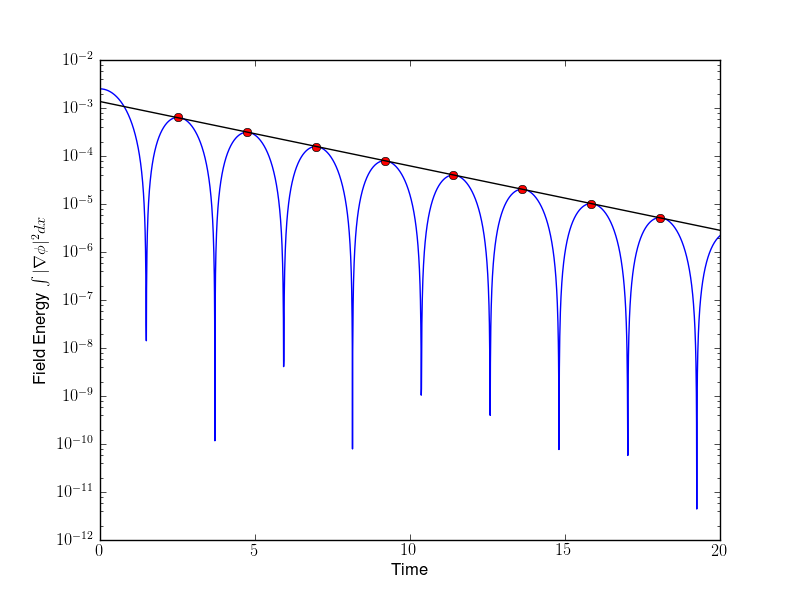}
  \incfig{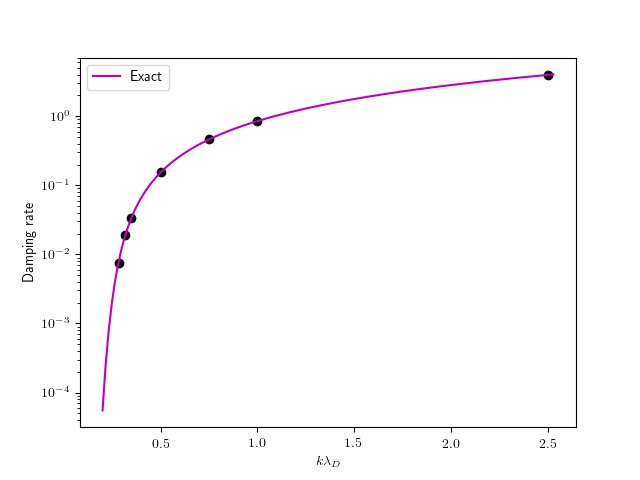}
  \caption{(Left) Field energy as a function of time for the linear
    Landau-damping problem with $k = 0.5$ and $T_e = 1.0$. The red
    dots represent the maxima in the field energy which are used to
    compute a linear least-squares fit. The slope of the black line
    gives the damping rate. (Right) Damping rate as a function of
    normalized Debye length. The black dots show the numerical damping
    rates compared to the exact results (magenta) computed from the
    dispersion relation.}
  \label{fig:land-damp}
\end{figure}

Landau damping is the damping process of electrostatic plasma waves
in a collisionless plasma. In this test the ability of the algorithm
to capture the phenomena of Landau damping is shown. For this the
electrons are initialized with a perturbed Maxwellian given by
\begin{align}
    f(x,v,0) = \frac{1}{\sqrt{2\pi v_t^2}} \exp(-v^2/2v_t^2)
  [1+\delta\cos(kx)],
\end{align}
where $v_t$ is the thermal velocity, $k$ is the wave number, and
$\delta$ controls the perturbation. Periodic boundary conditions are
imposed in the spatial direction and zero-flux conditions in the
velocity direction. The ion density is set to $1$. The electrostatic field is determined from
\begin{align}
    \frac{\partial^2 \phi}{\partial x^2} = -\frac{\varrho_c}{\epsilon_0} = -\frac{1}{\epsilon_0}\left(1-\int f dv\right)  
\end{align}

For all
these tests $m=1$, $\epsilon_0=1$ and $\delta=0.01$. With these
settings, the plasma frequency is $\omega_{pe}=1$ and the Debye length
is $\lambda_D = \sqrt{T_e}$. The wave number is varied and the damping
rates are computed as the slope of the least-squares line passing
through successive maxima of the field energy. See
Fig.\,\ref{fig:land-damp} for details, which shows the field energy
for the case $T_e=1.0$.  Figure\,\ref{fig:land-damp} also shows the
numerical results compared to the exact values obtained from the
dispersion relation for Langmuir waves\footnote{For plasma oscillations the
  dispersion relation is
  \begin{align*}
    1 - \frac{1}{2 k^2 \lambda_D^2} Z'(\zeta) = 0
  \end{align*}
  where $\lambda_D$ is the Debye length and
  \begin{align*}
    Z(\zeta) = \frac{1}{\sqrt{\pi}}
    \int_{-\infty}^\infty
    \frac{e^{-x^2}}{x-\zeta}
    dx
  \end{align*}
}.

\begin{figure}
  \setkeys{Gin}{width=1.0\linewidth,keepaspectratio}
  \incfig{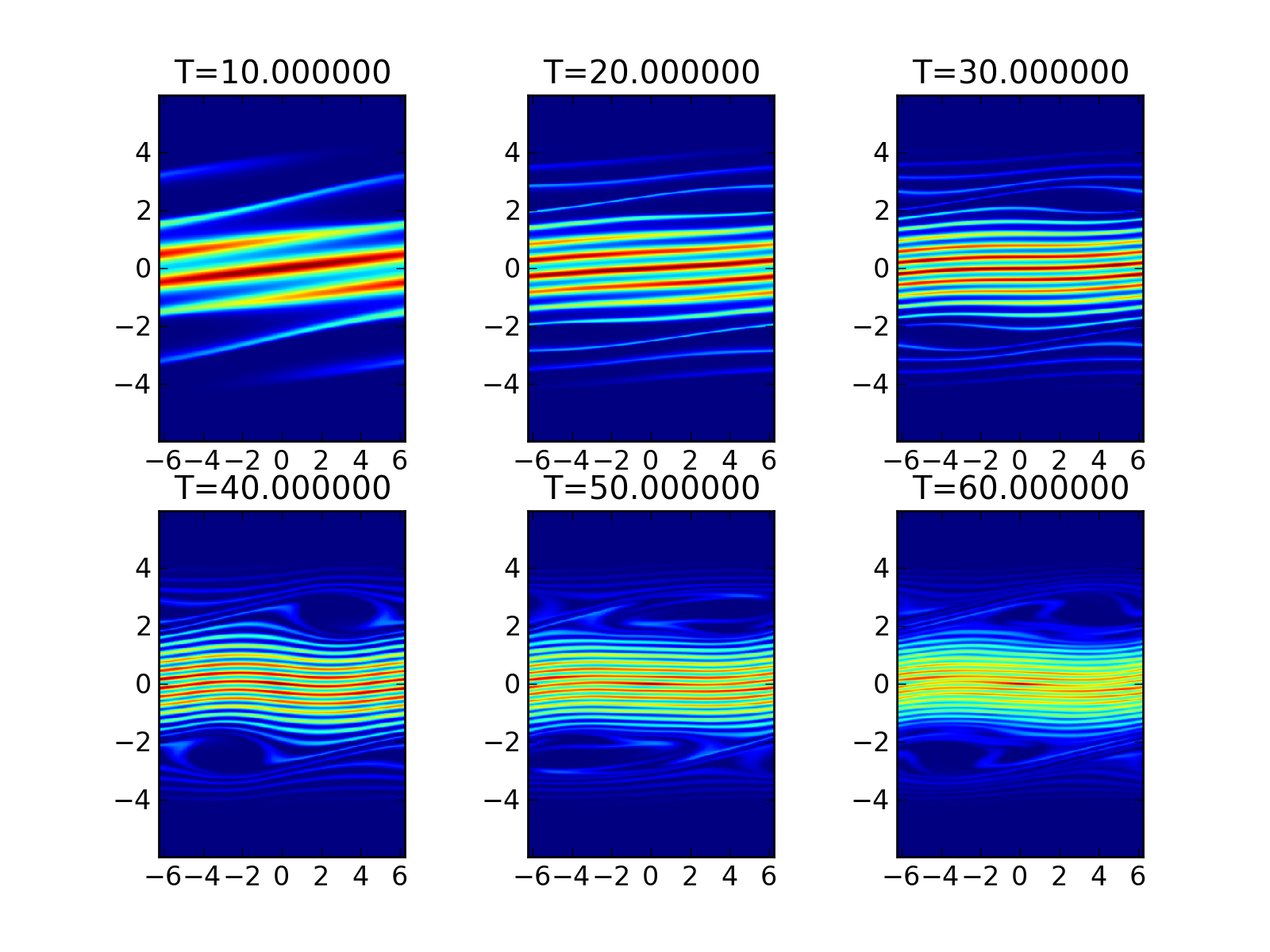}
  \caption{Distribution function at different times for the nonlinear
    Landau-damping problem. The initial perturbation undergoes
    shearing in phase space, leading to Landau damping from the phase
    mixing (see Fig.\,\ref{fig:land-damp} for the damping rate). Starting at around
    $t=20$ the damping is halted due to particle trapping, finally
    leading to saturation. Phase-space holes are clearly visible.}
  \label{fig:nonlin-land-damp}
\end{figure}

If we increase the perturbation to $\delta = 0.5$, the system is
rapidly driven nonlinear. Other parameters are kept the same as for
the linear Landau-damping problem with $k=0.5$ and $T_e=1.0$. In the
nonlinear phase, the Landau damping eventually halts due to the
formation of phase-space holes that lead to particle trapping that
``shuts off'' phase mixing (see Fig.\thinspace\ref{fig:nonlin-land-damp}).

Besides damping of electrostatic plasma waves, ion acoustic waves are also Landau damped. In this case we can assume that the electrons
are a massless isothermal fluid. Hence we
use, instead of the Poisson equation, the electron momentum equation
to determine the potential. That is
\begin{align}
  n_{i}(x) = n_{eo}\left(1 + \frac{|e|\phi}{T_e}\right)
\end{align}
where, in this linear probem, $n_{eo}$ is the constant electron
initial density and $T_e$ is the fixed electron temperature. This
allows the determination of the potential once the ion number density
is known. Note that this is an \emph{algebraic} equation for $\phi$
and not a PDE. However, as the number density that appears in this is
discontinuous (as the distribution function is discontinuous) this
expression would mean that the potential $\phi$ is also discontinuous,
violating the requirement for energy conservation. Hence, one realizes
that this equation can only hold \emph{weakly}, that is, one must
construct the potential such that it belongs to the continuous
subspace of the discontinuous basis: we must find $\phi_h$ such that
\begin{align}
  \int \psi \left[
  n_{i}(x) - n_{eo}\left(1 + \frac{|e|\phi_h}{T_e} \right)
  \right]\thinspace dx
  = 0
\end{align}
for all continuous test functions $\psi(x)$. This is a \emph{global}
equation for $\phi$ (analogous to the Poisson equation which is also
global) that results in a linear system needing inversion to compute
the discrete potential.

We have performed tests for Landau damping of ion-acoustic waves in this
approximation. The damping rates computed from our simulations (not
shown) match the exact rates very accurately. In fact, the plot of
damping with various values of $T_i/T_e$ looks identical to the right
panel of Fig.\thinspace\ref{fig:land-damp} as the dispersion relation
is exactly analogous to the case of Landau damping of electron
oscillations. For a nonlinear application in gyrokinetics
see\cite{Shi2015} in which we studied the problem of heat-flux on
divertor plates using a 1D gyrokinetic model, in which the field equation is local but the discrete potential needs to be determined using a non-local inversion along the magnetic field-line.

\subsection{Incompressible Euler flow}

In the final set of test problems we look at the incompressible Euler
equations instead of the Vlasov--Poisson equations. The algorithms
presented here are also applicable for this system as the
incompressible Euler equations in 2D can be written in Hamiltonian
form as described in the previous sections.

In the problem the simulation is initialized with two shear
layers. The initially planar shear layers are perturbed slightly due
to which they roll around each other, forming increasingly finer
vortex-like features. The initial conditions for this problem are
\begin{align}
  f(x,y,0) = 
  \left\{
    \begin{array}{l l}
      \delta\cos(x) - \mathrm{sech}^2[(y-\pi/2)/\rho]/\rho \quad y\le\pi \\
      \delta\cos(x) + \mathrm{sech}^2[(3\pi/2-y)/\rho]/\rho \quad y>\pi
    \end{array}
  \right..
\end{align}

\begin{figure}
  \setkeys{Gin}{width=1.0\linewidth,keepaspectratio}
  \incfig{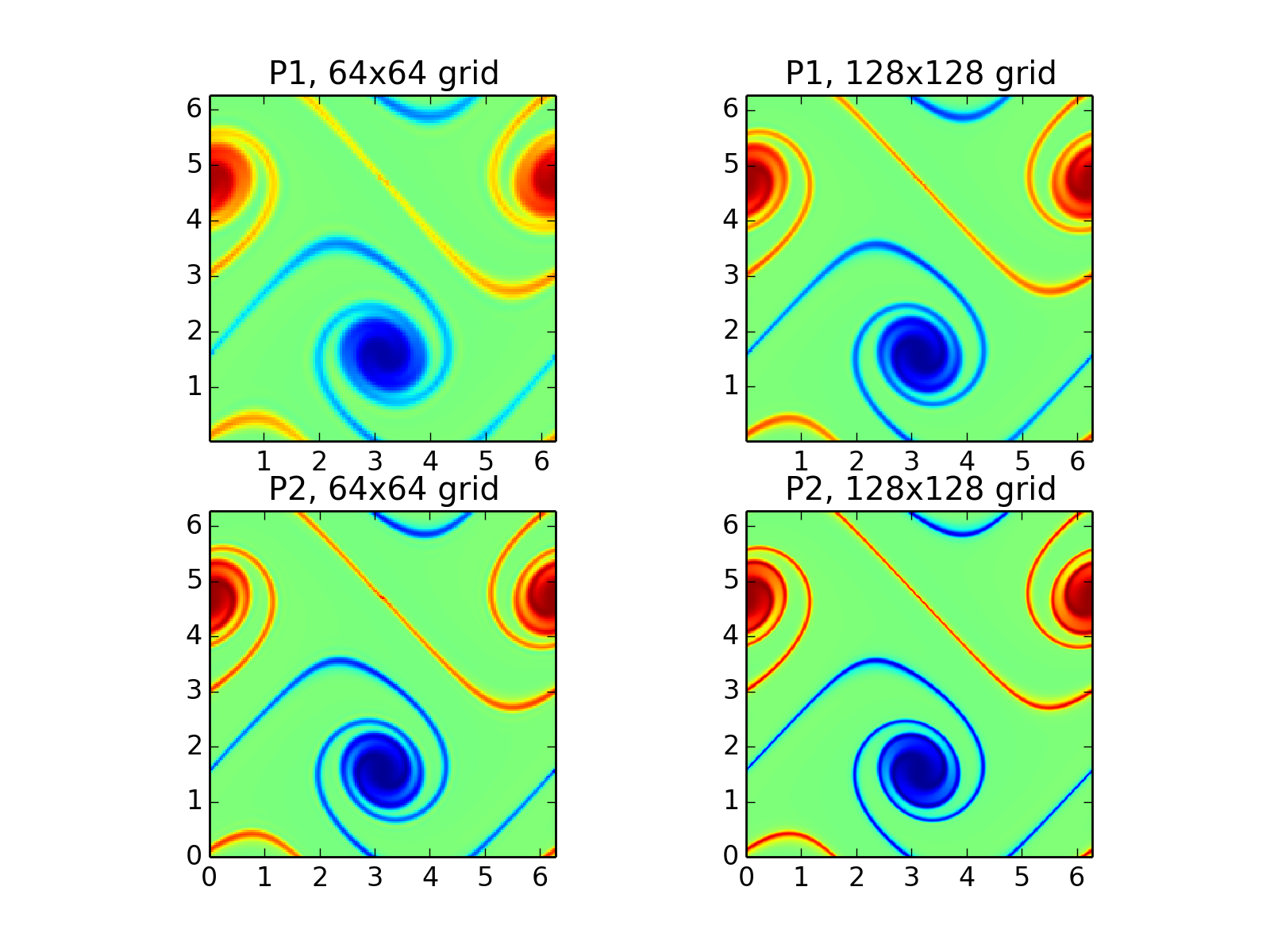}
  \caption{Double-shear-problem vorticity at t=8 with different grid
    resolutions and polynomial orders. Here, an upwind flux is used.}
  \label{fig:double-shear-upwind}
\end{figure}

For the results shown, $\rho = \pi/15$ and $\delta = 0.05$, and the problem is run to $t=8$. In the first set of simulations, an upwind flux was
used with different grid sizes and spatial-order schemes to compute
the solution. Figure\thinspace\ref{fig:double-shear-upwind} shows the
results at the final time from these simulations. Recall that
even though the energy is conserved with upwind fluxes, the enstrophy
is not. To conserve enstrophy one can use central fluxes.
With this the lack of energy and enstrophy conservation is solely due
to the damping
from the third-order Runge--Kutta time stepper used to advance the solution in
time. As seen in Fig.\thinspace\ref{fig:energy-enstrophy-hist}
these errors reduce as $O(\Delta t^3)$ as the time step is reduced.

\begin{figure}
  \incfig{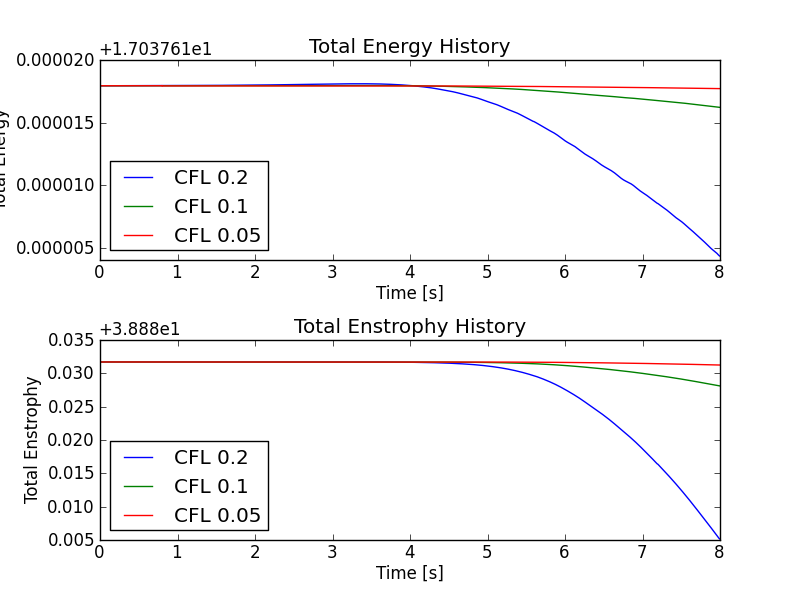}
  \caption{Total energy (top) and total enstrophy (bottom) history
    with different CFL numbers with central fluxes. Both energy and
    enstrophy errors asymptote to zero with the order of time-stepping
    scheme. The drop in energy is $1.36\times 10^{-5}$,
    $1.73\times 10^{-6}$ and $2.29\times 10^{-7}$ respectively, giving
    an order of $2.97$ and $2.91$. The drop in enstrophy
    is $2.66\times 10^{-2}$, $3.59\times 10^{-3}$ and
    $4.578\times 10^{-4}$ respectively, giving an order of $2.88$ and
    $2.97$.}
  \label{fig:energy-enstrophy-hist}  
\end{figure}

\begin{figure}
  \incfig{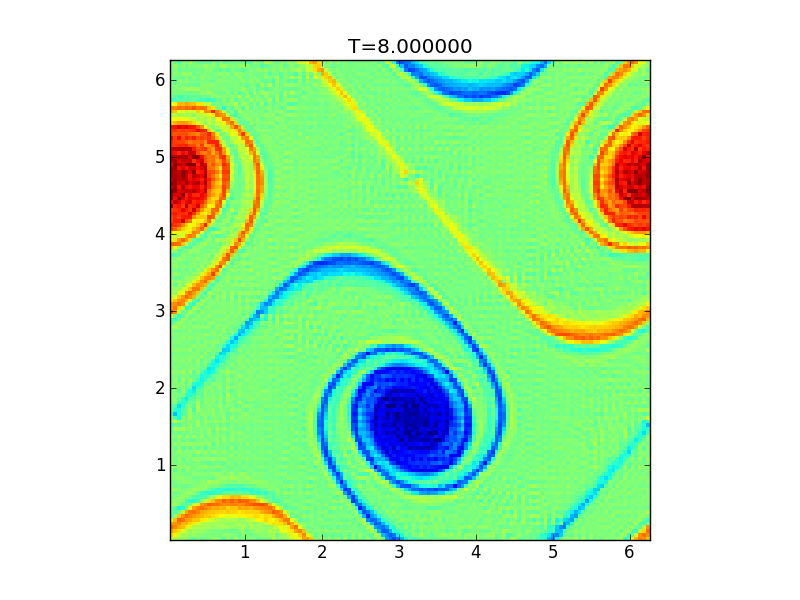}
  \caption{Vorticity at $t=8$ for double shear problem with central
    fluxes. Notice the significant phase errors in the solution as
    compared to the solution with the upwind fluxes (see Fig.\,\ref{fig:double-shear-upwind}). In general,
    a decaying enstrophy is desirable for stability.}
  \label{fig:double-shear-central}
\end{figure}

Figure\thinspace\ref{fig:double-shear-central} shows that
enstrophy conservation does not come for free: the lack of diffusion causes
significant phase errors when the flow structures reach the grid
scale. In general, a decaying enstrophy (or, more generally, a decaying $L_2$-norm of the solution) is desirable for stability.

\begin{figure}
  \setkeys{Gin}{width=1.0\linewidth,keepaspectratio}
  \incfig{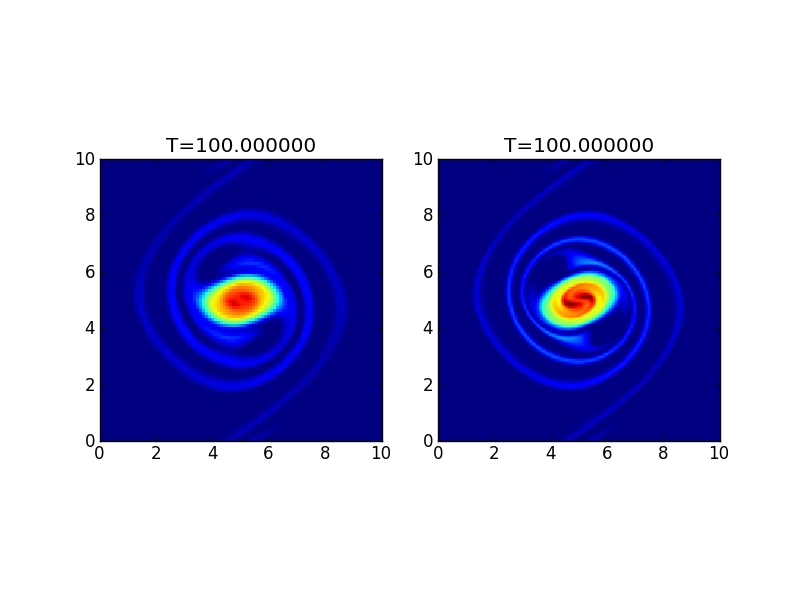}
  \caption{Vorticity for the vortex waltz problem with the
  second-order SP2 scheme. The left panel shows the solution with a $32 \times 32$
    grid, while the right panel shows the solution with a $64 \times 64$
    grid.}
  \label{fig:vortex-waltz}
\end{figure}

In the final problem the simulation was initialized with two Gaussian
vortices which merge as they orbit around each other. The vorticity is
initialized using the sum of two Gaussians given by
\begin{align}
  f(x,y,0) = \omega_1(x,y) + \omega_2(x,y)
\end{align}
where 
\begin{align}
  \omega_i(x,y) = e^{-r_i^2/0.8}
\end{align}
where $r_i^2 = (x-x_i)^2 + (y-y_i)^2$ and $(x_1,y_1) = (3.5,5.0)$ and
$(x_2,y_2) = (6.5,5.0)$. This is the ``vortex-waltz'' problem often
used to benchmark schemes for the incompressible Euler
equations. Figure\thinspace\ref{fig:vortex-waltz} shows the solutions
on $32\times 32$ and $64\times 64$ grids using the polynomial order 2
scheme with upwind fluxes. As for the double-shear problem, the energy
and (not shown) are conserved to the order of the time-stepping
scheme. Interestingly, the $p=2$ (SP2) scheme runs faster and gives
higher accuracy on a coarser mesh than the $p=1$ (SP1) scheme. In
general, this is a trend in all the simulations we have performed in
which $p=2$ seems to be a ``sweet spot'' for accuracy and speed.

\section{Conclusions}

We have presented novel discontinuous Galerkin (DG) and continuous
Galerkin (CG) schemes for the solution of a broad class of physical
systems described by Hamiltonian evolution equations. 
These schemes
conserve energy in the time-continuous limit, and the
$L_2$ norm (and hence the entropy) when using central-fluxes to update
the surface terms in the DG part of the scheme. 
They also conserve particles for arbitrary time step.
Our energy
conservation proof, rather surprisingly, is independent of upwinding
(which adds dissipation in the DG update) while computing the
distribution function. The physical reason for this is that updwinding
is applied to the phase-space advection that happens to be
perpendicular to the constant energy contours in phase-space. With
updwinding, however, the $L_2$ norm decays monotonically, leading to a
stable scheme with robust properties, providing some damping of barely-resolved grid-scale structures. We have
applied these schemes to a broad variety of test problems and provided
references to applications of our schemes to electrostatic and
electromagnetic gyrokinetic equations, a system of considerable
complexity that evolves in 5D phase-space. The scheme presented here
does not ensure that the distribution function is positive,
however. We have developed a novel scheme based on exponential
reconstruction to ensure positivity but without changing the
conservation properties of the scheme. This will be presented in a
future publication. Overall, the scheme presented here has proved to
be robust, accurate and, with novel choices of basis functions and
automated code generation \cite{HakimJuno2019}, very efficient. It
forms the basis of much of our research into plasma turbulence in
fusion machines \cite{mandell-2019,Shi2017thesis,Shi2017,Shi2019}.

\appendix

\vspace{1em}
{\bf Appendix:  Getting \gke\ and reproducing the results.}
%
To allow interested readers to reproduce our results and also use \gke\ for their applications, in this Appendix we provide instructions to get the code (in both binary and source format) as well as input files used here. Full installation instructions for \gke\ are provided on the \gke\ website (\url{http://gkeyll.readthedocs.io}). The code can be installed on Unix-like operating systems (including Mac OS and Windows using the Windows Subsystem for Linux) either by installing the pre-built binaries using the conda package manager (\url{https://www.anaconda.com}) or building the code via sources. The input files used here are under version control and can be obtained from the first author. Some are used as ``regression tests'' for the code. All tests can be run on a laptop.

\section*{Acknowledgments}

We are grateful for insights from conversations with Petr Cagas, Tess
Bernard, Jimmy Juno and other members of the \gke\ team. A. Hakim and G. Hammett are supported by 
the High-Fidelity Boundary Plasma Simulation SciDAC Project, part of
the DOE Scientific Discovery Through Advanced Computing (SciDAC)
program, through the U.S. Department of Energy contract
No. DE-AC02-09CH11466 for the Princeton Plasma Physics
Laboratory. A. Hakim is also supported by Air Force Office of
Scientific Research under Grant No. FA9550-15-1-0193.  N.~Mandell is
supported by the DOE CSGF program, provided under grant 
DE-FG02-97ER25308.

\bibliographystyle{siamplain}
\bibliography{gke}

\end{document}